\documentclass[12pt]{article} 

\usepackage[ngerman,english]{babel}
\useshorthands{"}
\addto\extrasenglish{\languageshorthands{ngerman}}
\usepackage[utf8]{inputenc}
\usepackage[T1]{fontenc}
\usepackage{amsmath}
\usepackage{graphicx}
\usepackage{amssymb}
\usepackage{hyperref}
\usepackage{amsmath}
\usepackage{mathrsfs}
\usepackage{color}
\usepackage{subfig}
\usepackage{amsthm}

\newtheorem{lemma}{Lemma}[section]

\newtheorem{theorem}{Theorem}[section]

\newtheorem{corollary}{Corollary}[section]

\theoremstyle{definition}
\newtheorem{example}{Example}[section]

\renewcommand{\theta}{\vartheta}

\title{{\Large\bf Formulation and Proof of the Gravitational Entropy Bound}}
\author{{\bf Artem Averin$^{\textrm{a}}$\footnote{artem.averin@campus.lmu.de}}}

\begin{document}

\maketitle

\centerline{\it $^{\textrm{a}}$ Arnold--Sommerfeld--Center for Theoretical Physics,}
\centerline{\it Ludwig--Maximilians--Universit\"at, 80333 M\"unchen, Germany}

\vskip1cm
\begin{abstract}
{{
We provide a formulation and proof of the gravitational entropy bound. We use a recently given framework which expresses the measurable quantities of a quantum theory as a weighted sum over paths in the theory's phase space. If this framework is applied to a field theory on a spacetime foliated by a hypersurface $\Sigma,$ the choice of a codimension-2 surface $B$ without boundary contained in $\Sigma$ specifies a submanifold in the phase space. We show here that this submanifold is naturally restricted to obey an entropy bound if the field theory is diffeomorphism-invariant. We prove this restriction to arise by considering the quantum-mechanical sum of paths in phase space and exploiting the interplay of the commutativity of the sum with diffeomorphism-invariance. The formulation of the entropy bound, which we state and derive in detail, involves a functional $K$ on the submanifold associated to $B.$ We give an explicit construction of $K$ in terms of the Lagrangian. The gravitational entropy bound then states: For any real $\frac{\lambda}{\hbar},$ consider the set of states where $K$ takes a value not bigger than $\lambda$ and let $V$ denote the phase space volume of this set. One has then $\ln (V) \le \frac{\lambda}{\hbar}.$ Especially, we show for the Einstein-Hilbert Lagrangian in any dimension with cosmological constant and arbitrary minimally coupled matter, one has $K = \frac{A}{4G}.$ Hereby, $A$ denotes the area of $B$ in a particular state.                  
}}
\end{abstract}


\newpage

\setcounter{tocdepth}{2}
\tableofcontents
\break

\section{Introduction}
\label{Kapitel 1}

We give a rigorous formulation and proof of the gravitational entropy bound and in the technical arguments to be presented in the next chapters, we heavily rely here on the framework presented in \cite{Averin:2024its}. After providing a brief introduction to the gravitational entropy bound and pointing out the need for its formulation in this chapter, we expect the reader in the next chapters to be familiar with the formalism, language and notation of all chapters of \cite{Averin:2024its} (the chapter 4.1 presenting an example may be skipped). 

The qualitative statement of the gravitational entropy bound is motivated by black hole thermodynamics. In the following, we briefly review some of the related statements and arguments that are going to be relevant for us. They are well-known and covered by typical introductory textbooks (see, for instance, \cite{Kiritsis:2019npv} which also contains a list of the original references). 

Naively, one expects the entropy $S$ inside a region to be bounded by the entropy $S_{BH}$ of a black hole placed in this region. The gravitational entropy bound is hence expected to be of the form $S \le S_{BH}.$ Since the black hole entropy scales as an area-law in Einstein-gravity, this qualitative expectation is also known as the holographic entropy bound.

The problem with the expected bound $S \le S_{BH}$ is apparent. It is not clear how to define precisely each of its ingredients. What is an invariant meaning to fix a region? Then, what should its entropy $S$ be? And what is $S_{BH}$ for an arbitrary region? 

A proposal was conjectured in \cite{Bousso:1999xy}. Here, we follow a different approach. We use as motivation the formulation and proof given in \cite{Casini:2008cr} of the related Bekenstein bound.

Qualitatively, the Bekenstein bound $S \lesssim \frac{ER}{\hbar}$ was proposed to bound the entropy of a system with energy $E$ inside a region of size $R$ in any quantum field theory. Indeed, a precise notion of this was proven in \cite{Casini:2008cr}. What in the argument is important for us, is that for any Lorentz invariant quantum field theory a specifically reduced density matrix of the ground state can be given explicitly. This is due to the Lorentz invariance which fixes its form. This reduced density matrix then implies an entropy bound which materializes the qualitatively expected $S \lesssim \frac{ER}{\hbar}.$

Here, we want to generalize this to field theories with diffeomorphism invariance. We ask whether the additional symmetry imposes additional restrictions. For instance, is it possible to deduce some sort of entropy bound by similar reasoning as in the derivation of the Bekenstein bound?

Indeed, our result here is that this is the case for arbitrary diffeomorphism invariant field theories. \cite{Averin:2024its} precisely gives the appropriate methods for the formulation and derivation. 

In \cite{Averin:2024its}, it is shown how quantum-mechanical observables are expressed as a weighted sum of paths in the phase space of a given theory. An important point is the manifest commutativity of this sum (i.e. covariance of the functional integral). 

Furthermore, the possifold-flow introduced in \cite{Averin:2024its} provides a notion how to reduce a given density matrix. The choice of a particular reduction is related to a choice of a submanifold in the theory's phase space. Such a choice can be made by fixing a particular codimension-2 surface $B$ on the theory's spacetime.

We will demonstrate here, that one can determine the reduced density matrix with the help of a specific functional $K$ which, as we show, can be given explicitly for each $B$ on the associated submanifold in phase space. It is the commutativity of the sum of paths in phase space together with the diffeomorphism invariance that fixes the form of the reduced density matrix this way. 

Finally, we will show this reduced density matrix to imply an entropy bound.

In particular, we show that for any $\frac{\lambda}{\hbar} \in \mathbb{R}$ the phase space volume $V$ of the inverse image $K^{-1}((-\infty,\lambda])$ is by Theorem \ref{Theorem 3.1} bounded as $\ln (V) \le \frac{\lambda}{\hbar}.$ 

This provides a realization of the bound $S \le S_{BH}$ asked for at the beginning of this chapter. It naturally gives a concrete meaning of each ingredient. The ``$S$'' is interpreted as the $\ln (V).$ We will see later on why the identification of ``$S_{BH}$'' with the functional $\frac{K}{\hbar}$ is justified.

To summarize, our main result is that one can make certain global statements about the phase space of any given diffeomorphism invariant field theory. This is achieved by some kind of ``integral geometry.'' We consider the quantum-mechanical sum over paths in phase space. The covariance of this sum together with diffeomorphism invariance imposes then restrictions on the phase space. We will show here that we can understand the origin of the gravitational entropy bound this way.

The paper is organized as follows. The gravitational entropy bound is presented in Theorem \ref{Theorem 3.1}. The proof is given in chapter \ref{Kapitel 2} and \ref{Kapitel 3}. It is portioned in several lemmas in order to increase clarity. Chapter \ref{Kapitel 2} focuses on the quantum-mechanical sum over paths. In chapter \ref{Kapitel 3} the mentioned functional $K$ is constructed by using the explicit symplectic structure of diffeomorphism invariant field theories. We end with the discussion in chapter \ref{Kapitel 4}.

We work in natural units although we will occasionally restore $\hbar$ (or $G$) in order to highlight the entrance of the relevant physics.      

\section{Derivation}
\label{Kapitel 2}

As stated in the beginning, we remind here that in all chapters to follow, we will use the framework of \cite{Averin:2024its} including formalism, language and notation. 

In the following chapters, we consider a theory $(\Gamma, \Theta, H)$ in the sense of chapter 2 of \cite{Averin:2024its} with corresponding Hilbert-space $\mathcal{H}.$ Furthermore, we choose in this theory a state of the form 

\begin{equation}
| \Psi \rangle = e^{HT_E} | \chi \rangle
\label{1}
\end{equation}

for some fixed Euclidean time $T_E \in \mathbb{R}$ and a position-eigenstate $| \chi \rangle.$ In what follows, this is not a restriction on the state $| \Psi \rangle$ as every state can be written as a linear combination of states of the form \eqref{1} for some fixed time $T_E.$ 

The following lemma gives a functional integral representation of the density matrix associated to the state \eqref{1} using the notation of Appendix A in \cite{Averin:2024its}: 

\begin{lemma} \label{Lemma 2.1}
For position-eigenstates $|q_A \rangle, |q_B \rangle \in \mathcal{H},$ the density matrix elements of \eqref{1} are given by

\begin{equation} \label{2}
\begin{split}
& \langle q_A | \Psi \rangle \langle \Psi | q_B \rangle \\
=&\int_{q(-iT_E) = \chi}^{q(0)=q_A} {\operatorname{Vol}(t) e^{i\int_{-iT_E}^{0}dt \left( \Theta_{\Phi(t)} [\dot{\Phi}(t)] - H[\Phi(t)] \right)}} \cdot \\
& \int_{q(0)=q_B}^{q(iT_E)=\chi} {\operatorname{Vol}(t) e^{i\int_{0}^{iT_E}dt \left( \Theta_{\Phi(t)} [\dot{\Phi}(t)] - H[\Phi(t)] \right)}}.
\end{split}
\end{equation}

\end{lemma}

\begin{proof}
The position-wavefunction of \eqref{1} can be represented as 

\begin{equation} \label{3}
\begin{split}
& \langle q | \Psi \rangle \\
=& \langle q | e^{-iH(iT_E)} | \chi \rangle \\
=& \int_{q(0)=\chi}^{q(iT_E)=q} \operatorname{Vol}(t) e^{i\int_{0}^{iT_E}dt \left( \Theta_{\Phi(t)} [\dot{\Phi}(t)] - H[\Phi(t)] \right)}.
\end{split}
\end{equation}

Analogously, 

\begin{equation} \label{4}
\begin{split}
&\langle \Psi | q \rangle = \langle \chi | e^{HT_E} | q \rangle = \langle \chi | e^{-iH(iT_E)} | q \rangle \\
=& \int_{q(0)=q}^{q(iT_E)=\chi} \operatorname{Vol}(t) e^{i\int_{0}^{iT_E}dt \left( \Theta_{\Phi(t)} [\dot{\Phi}(t)] - H[\Phi(t)] \right)}.
\end{split}
\end{equation}

Performing the change of variables $\tilde{t} = t-iT_E$ in \eqref{3} and combining these expressions, we obtain

\begin{equation} \label{5}
\begin{split}
& \langle q_A | \Psi \rangle \langle \Psi | q_B \rangle \\
=& \int_{q(-iT_E)=\chi}^{q(0)=q_A} \operatorname{Vol}(\tilde{t}) e^{i\int_{-iT_E}^{0}d\tilde{t} \left( \Theta_{\Phi(\tilde{t})} \left[\frac{\delta \Phi}{d \tilde{t}} \right] - H \left[\Phi(\tilde{t})\right] \right)} \cdot \\
& \int_{q(0)=q_B}^{q(iT_E)=\chi} \operatorname{Vol}(t) e^{i\int_{0}^{iT_E}dt \left( \Theta_{\Phi(t)} [\dot{\Phi}(t)] - H[\Phi(t)] \right)} \\
=& \int_{q(-iT_E)=\chi}^{q(0)=q_A} \operatorname{Vol}(t) e^{i\int_{-iT_E}^{0}dt \left( \Theta_{\Phi(t)} [\dot{\Phi}(t)] - H[\Phi(t)] \right)} \cdot \\
& \int_{q(0)=q_B}^{q(iT_E)=\chi} \operatorname{Vol}(t) e^{i\int_{0}^{iT_E}dt \left( \Theta_{\Phi(t)} [\dot{\Phi}(t)] - H[\Phi(t)] \right)}
\end{split}
\end{equation}

where the last equation follows by renaming the time variable in the first integral and proves the assertion. 
\end{proof} 

The integration in Lemma \ref{Lemma 2.1} is visualized in Fig. \ref{Fig. 1}. 

\begin{figure}[h!]
\centering
  \includegraphics[trim = 0mm 90mm 0mm 30mm, clip, width=0.7\linewidth]{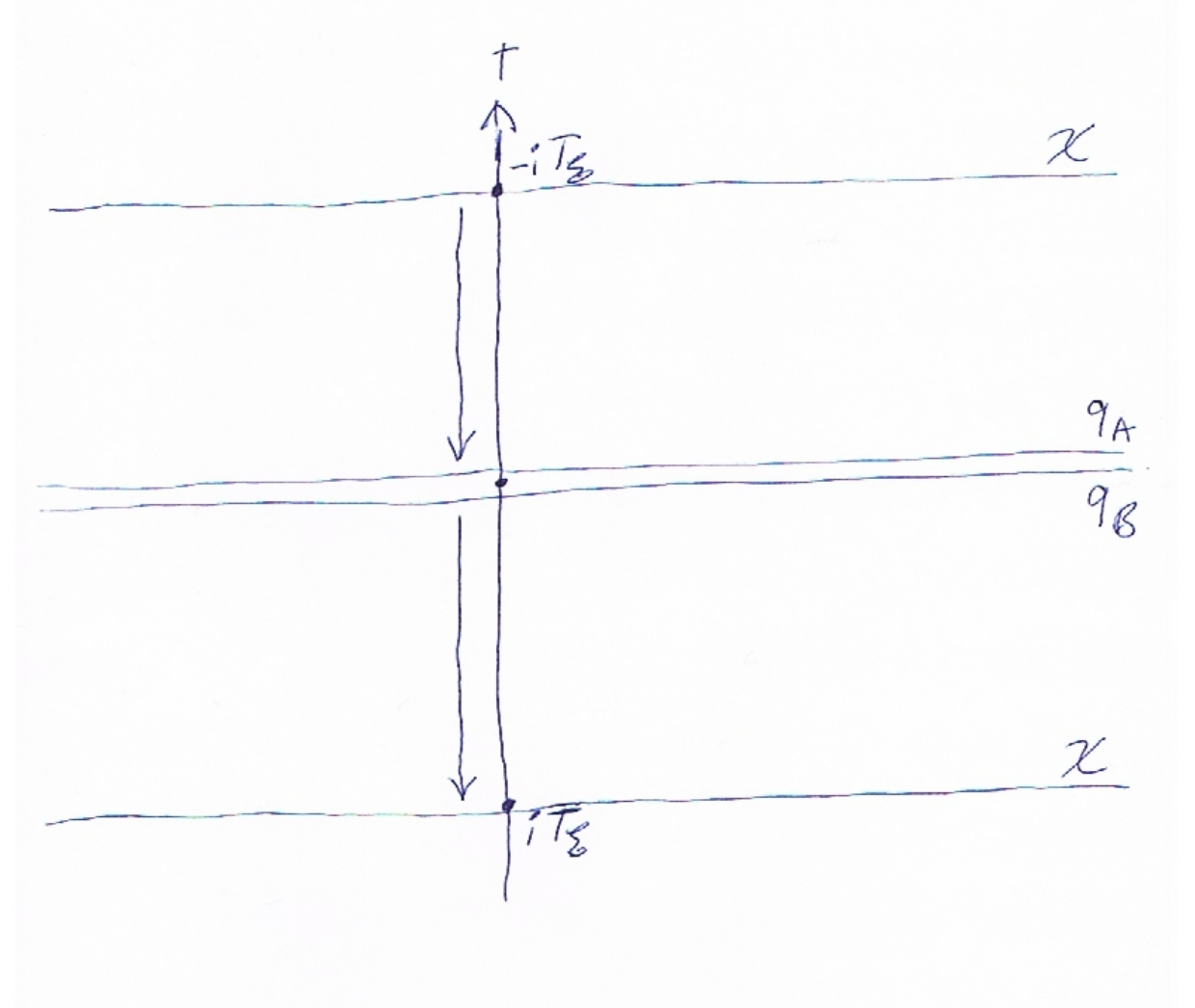}
  \caption{Pictorial representation of the integration \eqref{2}. Note the different states $q_A$ and $q_B$ at the time slice $t=0$ in each integration.}
	\label{Fig. 1}
\end{figure}

Having a functional integral expression for the density matrix, the next step is clear as motivated in the introduction. We want to use the covariance of the functional integral in order to bring the expression for the reduced density matrix into a desirable form.

Immediately, two questions appear. What is a sensible way to reduce the density matrix, i.e. how can the Hilbert-space $\mathcal{H}$ be decomposed into a tensor product? And, how can the covariance of the functional integration be exploited in a useful way?

Both questions have natural answers if one is dealing with a field theory as we have already partially encountered in \cite{Averin:2024its}. 

In addition to the requirements stated at the beginning of this chapter, we assume from now for the remaining part of the text $(\Gamma, \Theta, H)$ to be a field theory. The meaning of this was explained in detail below Fig. 3 in \cite{Averin:2024its}. We precisely require the situation stated there. That is, $(\Gamma, \Theta, H)$ is a field theory on spacetime $M = \mathbb{R} \times \Sigma$ with an associated action and Lagrange-form $L$ as in equation $(12)$ of \cite{Averin:2024its}.

Then, in this context, the answer to the first question above is given by the possifold-flow defined in \cite{Averin:2024its}. As in equation $(32)$ of \cite{Averin:2024its}, we consider in the following a possifold-flow $\partial \Sigma \to B$ for a codimension-2 surface $B$ bounding a subset $\Sigma_1 \subseteq \Sigma,$ i.e. $\partial \Sigma_1 = B,$ and $\Sigma_2$ being its complement $\Sigma_2 = \Sigma \setminus \Sigma_1.$ As explained in \cite{Averin:2024its}, the Hilbert-space then factorizes and a position-eigenstate is then of the form

\begin{equation}
|q \rangle = |q_1 \rangle |q_2 \rangle
\label{6}
\end{equation}

with $q_1$ being generalized coordinates on $\Sigma_1$ and $q_2$ being generalized coordinates on $\Sigma_2,$ respectively.

We will use the factorization \eqref{6} in order to reduce the density matrix \eqref{2}. Note that such a reduction of the density matrix naturally depends on the choice of a codimension-2 surface $B.$ 

Our explanations answer the first question raised above. The second question about the concrete use of the covariance of the functional integration is answered by the next lemma.

The idea is again natural from the field theory point of view. The spacetime $M$ is foliated by the hypersurface $\Sigma$ along the time-slices labeled by $t.$ However, it appears that a different foliation of $M$ by a hypersurface $\tilde{\Sigma}$ along a time direction $\tilde{t}$ should lead to the same results. This is precisely the content of the following lemma and the situation is illustrated in Fig. \ref{Fig. 2}.

\begin{figure}[h!]
\centering
  \includegraphics[trim = 0mm 90mm 0mm 20mm, clip, width=0.7\linewidth]{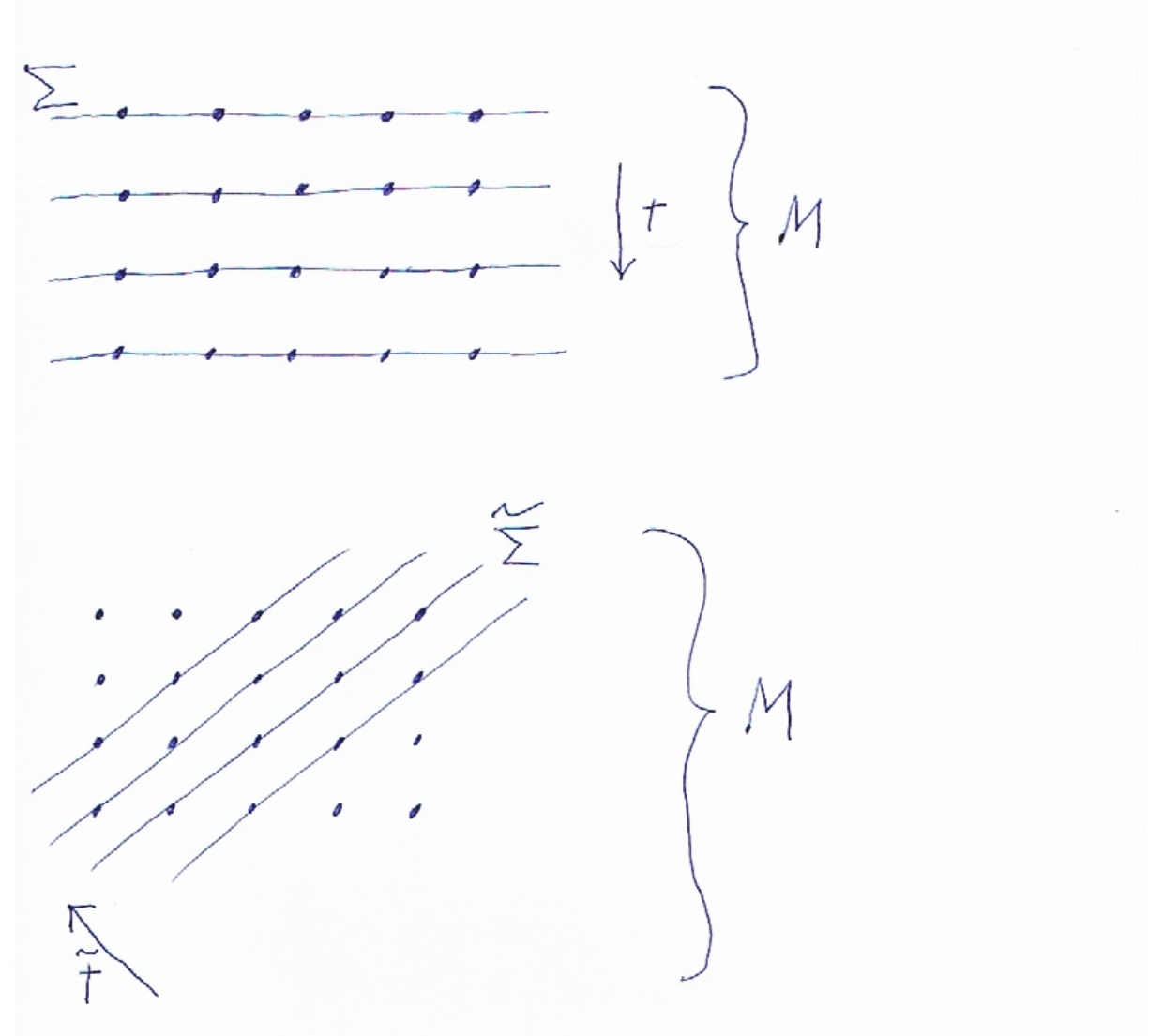}
  \caption{Different foliations of the spacetime $M.$ The spacetime $M$ is represented by a lattice before the continuum limit is taken in $\Sigma$ and in $t.$}
	\label{Fig. 2}
\end{figure}

\begin{lemma} \label{Lemma 2.2}
Let $\tilde{\Sigma}$ be a hypersurface foliating $M$ along the time-slices $\tilde{t}$ (as in Fig. \ref{Fig. 2}). The orientation of $\Sigma$ and $\tilde{\Sigma}$ should furthermore be induced by the time-directions $t$ and $\tilde{t}.$ Then, 

\begin{equation} \label{7}
\begin{split}
& \int \operatorname{Vol}(t) e^{i\int dt \left( \Theta_{\Phi(t)} [\dot{\Phi}(t)] - H[\Phi(t)] \right)} \\
=& \int \widetilde{\operatorname{Vol}}(\tilde{t}) e^{i\int d\tilde{t} \left( \tilde{\Theta}_{\Phi(\tilde{t})} \left[ \frac{\delta \Phi}{d \tilde{t}} \right] - K[\Phi(\tilde{t}),\tilde{t}] \right)}.
\end{split}
\end{equation}

Hereby,

\begin{equation} \label{8}
\begin{split}
\Theta &= \int_{\Sigma} \theta \\
\tilde{\Theta} &= \int_{\tilde{\Sigma}} \theta,
\end{split}
\end{equation}

where $\theta$ is the presymplectic potential.\footnote{See the explanations of equation $(16)$ in \cite{Averin:2024its} for the details especially concerning the meaning and fixing of ambiguities.} The functionals $H$ and $K$ are the generators of the time-evolution of $\Sigma$ and $\tilde{\Sigma}$ along $t$ and $\tilde{t}.$ $\operatorname{Vol}$ and $\widetilde{\operatorname{Vol}}$ are the associated volume forms to $\Theta$ and $\tilde{\Theta}.$
\end{lemma}

\begin{proof}
Let $(t,x)$ and $(\tilde{t}, \tilde{x})$ be two positively oriented coordinate systems of $M$ with $x$ and $\tilde{x}$ denoting the coordinates on $\Sigma$ and $\tilde{\Sigma}$ for fixed time $t$ and $\tilde{t},$ respectively. In this proof, we refer to the notation of the discussion below Fig. 3 of \cite{Averin:2024its}. Especially, we imagine the spacetime $M$ occasionally as a lattice as in Fig. \ref{Fig. 2} in order to make the point clear.

We can apply Jacobi's theorem to the action integral

\begin{equation}
\int dt dx \mathcal{L} = \int d\tilde{t} d\tilde{x} \tilde{\mathcal{L}}
\label{9}
\end{equation}

with $\mathcal{L}$ and $\tilde{\mathcal{L}}$ being the Lagrangian density with respect to $(t,x)$ and $(\tilde{t},\tilde{x}).$ For both sides of the equation \eqref{9}, we can apply the argument of \cite{Averin:2024its} starting below equation $(13)$ and leading to equation $(16)$ there. The latter equation gives the canonical 1-form in each case.

For the left-hand-side of \eqref{9}, we obtain

\begin{equation}
\Theta = \sum_{x \in \Sigma} p_x \delta q_x = \int_\Sigma \theta
\label{10}
\end{equation}

where $q_x$ denote the generalized coordinates on $\Sigma$ and $p_x$ are the generalized momenta defined with respect to the evolution of $\Sigma$ along $t.$ 

Analogously, for the right-hand-side of \eqref{9}, we obtain

\begin{equation}
\tilde{\Theta} = \sum_{\tilde{x} \in \tilde{\Sigma}} \tilde{p}_{\tilde{x}} \delta \tilde{q}_{\tilde{x}} = \int_{\tilde{\Sigma}} \theta
\label{11}
\end{equation}

where $\tilde{q}_{\tilde{x}}$ denote the generalized coordinates on $\tilde{\Sigma}$ and $\tilde{p}_{\tilde{x}}$ are the generalized momenta defined with respect to the evolution of $\tilde{\Sigma}$ along $\tilde{t}.$ 

Given a path $(q_x(t),p_x(t)),$ we can think of it as a configuration of canonical coordinates in Fig. \ref{Fig. 2} where the evolution of the $(q,p)$ is considered on $\Sigma$ along $t.$ Similarly, a path $(\tilde{q}_{\tilde{x}} (\tilde{t}), \tilde{p}_{\tilde{x}} (\tilde{t}))$ can be interpreted as a configuration in Fig. \ref{Fig. 2} where the evolution of the $(\tilde{q}, \tilde{p})$ is considered along the foliation with $\tilde{\Sigma}.$ 

For both ways of the evolution, we can construct the action functionals. For the evolution of $\Sigma$, it reads

\begin{equation}
\int dt \sum_{i} p_i \frac{\delta q_i}{dt} - \int dt H[q_i(t),p_i(t)]
\label{12}
\end{equation}

while for $\tilde{\Sigma}$, it is

\begin{equation}
\int d \tilde{t} \sum_i \tilde{p}_i \frac{\delta \tilde{q}_i}{d \tilde{t}} - \int d \tilde{t} K[\tilde{q}_i (\tilde{t}), \tilde{p}_i (\tilde{t}),\tilde{t}].
\label{13}
\end{equation}

The functional $K$ is here the generator of the evolution of $\tilde{\Sigma}$ along $\tilde{t}.$ Because \eqref{12} and \eqref{13} are derived from the same action, for a path $(q_i(t),p_i(t))$, there is a path $(\tilde{q}_i (\tilde{t}), \tilde{p}_i (\tilde{t}))$ such that \eqref{12} and \eqref{13} are equal  

\begin{equation}
\begin{split}
& \int dt \sum_i p_i \frac{\delta q_i}{dt} - \int dt H[q_i(t),p_i(t)] \\
=& \int d \tilde{t} \sum_i \tilde{p}_i \frac{\delta \tilde{q}_i}{d \tilde{t}} - \int d \tilde{t} K[\tilde{q}_i (\tilde{t}), \tilde{p}_i (\tilde{t}),\tilde{t}].
\end{split}
\end{equation}

Using \eqref{10} and \eqref{11}, this equation shows the equality of the individual terms in \eqref{7} under the mentioned map of $(q_i(t),p_i(t))$ to $(\tilde{q}_i (\tilde{t}), \tilde{p}_i (\tilde{t}))$. Furthermore, since the sum over all positions and momenta is taken on each side of \eqref{7}, the equality follows.     
\end{proof}

Our strategy is to use Lemma \ref{Lemma 2.2} to rewrite the functional integral expression for the reduced density matrix implied by \eqref{2}. For this, we need to choose a proper foliation of spacetime such that the generator for the evolution along this foliation is known. Field theories especially suited for this purpose seem to be those possessing local symmetries. Here, we choose to focus on diffeomorphism invariance. 

Hence, from now on, we require the field theory $(\Gamma, \Theta, H)$ to be diffeomorphism invariant in the sense that the Lagrange-form $L$ be diffeomorphism-invariant. A precise definition and further analysis of this is given in \cite{Iyer:1994ys}. We will elaborate more on this requirement in the next chapter. For the arguments here, these details will not be important. 

For the mentioned possifold-flow $\partial \Sigma \to B$ and associated factorization \eqref{6}, the diffeomorphism invariance allows us to obtain an expression for the reduced density matrix as stated in the following lemma. 

\begin{lemma} \label{Lemma 2.3}
For the possifold-flow $\partial \Sigma \to B,$ there exists a generator $G$ on $(\Gamma^{(B)}, \Theta^{(B)})$ such that the reduced density matrix elements of the state $| \Psi \rangle$ in \eqref{1} are given by

\begin{equation} \label{14}
\begin{split}
&\int \mathcal{D} q_2 ( \langle q_{1,A} | \langle q_2 | ) |\Psi \rangle \langle \Psi | ( |q_{1,B} \rangle |q_2 \rangle ) \\
=& \int_{q_1 (\alpha = 0)=q_{1,B} \atop q(\alpha=\pi, t=\pm iT_E)=\chi}^{q_1(\alpha=2\pi)=q_{1,A}} \operatorname{Vol}^{(B)}(\alpha) e^{i\int_{0}^{2\pi}d \alpha \left( \Theta^{(B)}_{\Phi(\alpha)} \left[\frac{\delta \Phi}{d \alpha} \right] - G[\Phi(\alpha)] \right)}.
\end{split}
\end{equation}  
\end{lemma}

\begin{proof}
Using the factorization \eqref{6} and tracing over $q_2,$ we obtain from \eqref{2} for the reduced density matrix elements

\begin{equation} \label{15}
\begin{split}
&\int \mathcal{D} q_2 ( \langle q_{1,A} | \langle q_2 | ) |\Psi \rangle \langle \Psi | ( |q_{1,B} \rangle |q_2 \rangle ) \\
=& \int \mathcal{D} q_2 \int_{q(-iT_E) = \chi}^{q_1(0)=q_{1,A},q_2(0)=q_2} {\operatorname{Vol}(t) e^{i\int_{-iT_E}^{0}dt \left( \Theta_{\Phi(t)} [\dot{\Phi}(t)] - H[\Phi(t)] \right)}} \cdot \\
& \int_{q_1(0)=q_{1,B},q_2(0)=q_2}^{q(iT_E)=\chi} {\operatorname{Vol}(t) e^{i\int_{0}^{iT_E}dt \left( \Theta_{\Phi(t)} [\dot{\Phi}(t)] - H[\Phi(t)] \right)}}.
\end{split}
\end{equation}

The functional integration is visualized in Fig. \ref{Fig. 3}.

\begin{figure}[h!]
\centering
  \includegraphics[trim = 0mm 90mm 0mm 20mm, clip, width=0.7\linewidth]{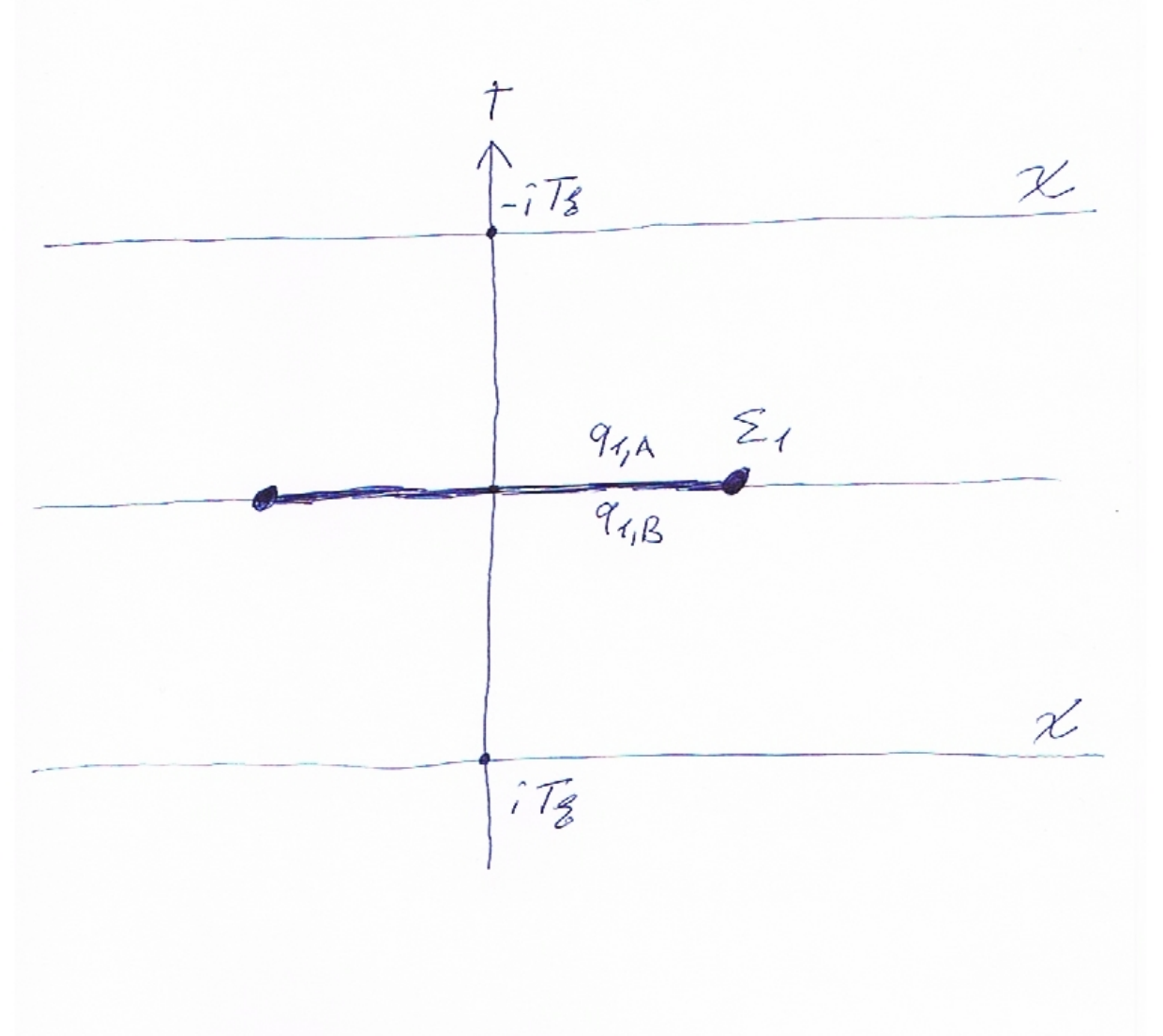}
  \caption{Functional integration \eqref{15}. The thick line denotes the surface $\Sigma_1$ at $t=0$ where the states are fixed as $q_{1,A}$ and $q_{1,B}.$ The thick endpoints illustrate the boundary $\partial \Sigma_1 = B.$}
	\label{Fig. 3}
\end{figure}

We now would like to apply Lemma \ref{Lemma 2.2} to this integration. For this, we need a suited new foliation of the spacetime in Fig. \ref{Fig. 3}. A natural choice is depicted in Fig. \ref{Fig. 4}. We choose to label the new foliation by the new time variable $\alpha \in [0, 2 \pi].$ 

\begin{figure}[h!]
\centering
  \includegraphics[trim = 0mm 90mm 0mm 20mm, clip, width=0.7\linewidth]{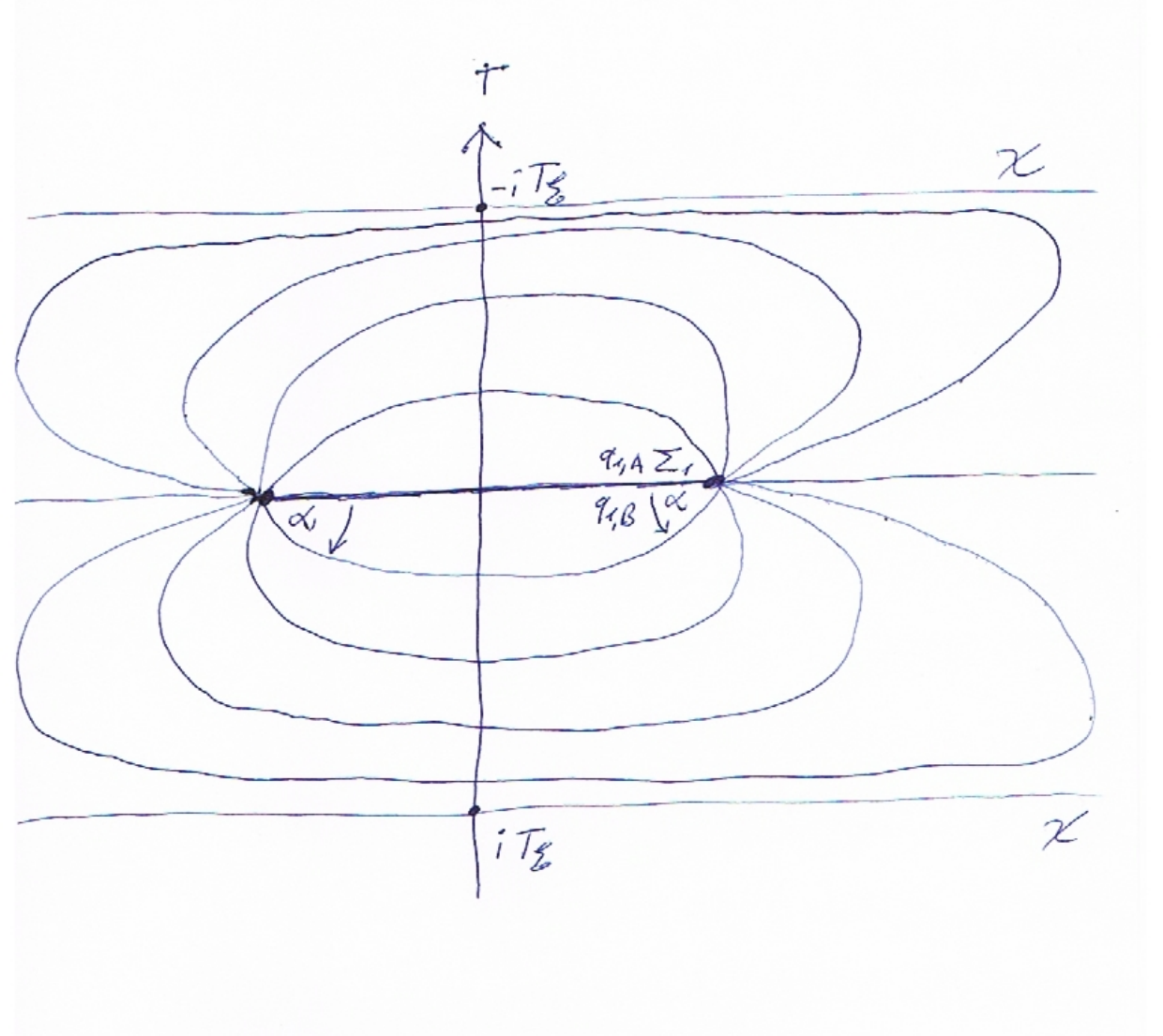}
  \caption{A different foliation for the functional integration \eqref{15} in Fig. \ref{Fig. 3}.}
	\label{Fig. 4}
\end{figure}

In order to apply Lemma \ref{Lemma 2.2}, we would need to find the generator $G$ corresponding to the evolution of the foliation shown in Fig. \ref{Fig. 4}. Here is where we can make use of the diffeomorphism invariance. 

Let $\Phi \in \Gamma^{(B)}$ be a state on $\Sigma_1.$ What is the role the generator $G$ should accomplish on $(\Gamma^{(B)}, \Theta^{(B)})$? Consider the situation in Fig. \ref{Fig. 5}.

\begin{figure}[h!]
\centering
  \includegraphics[trim = 0mm 90mm 0mm 20mm, clip, width=0.7\linewidth]{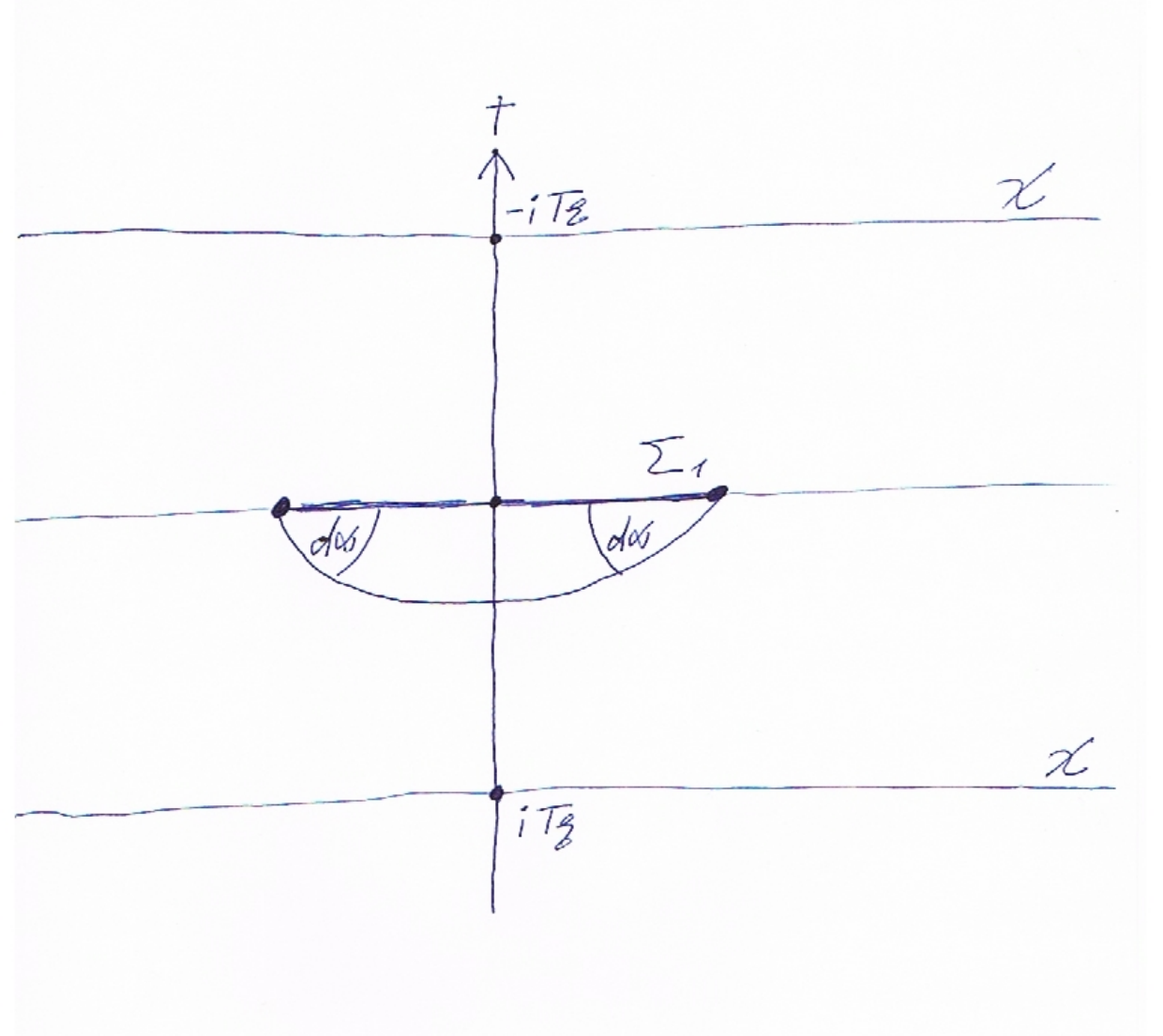}
  \caption{Action of the generator $G.$}
	\label{Fig. 5}
\end{figure}

In Fig. \ref{Fig. 5} an infinitesimal deformation of $\Sigma_1$ is shown. Such a deformation can be described by a suited diffeomorphism. Following the idea illustrated in Fig. \ref{Fig. 4}, $G$ should then be the generator of this diffeomorphism. 

However, in diffeomorphism-invariant field theories, the interior deformations in Fig. \ref{Fig. 5} correspond to interior changes of the diffeomorphism and hence to gauge redundancies. In other words, the specific form of the hypersurfaces in their interior in Fig. \ref{Fig. 4} does not matter. 

Following Fig. \ref{Fig. 4}, we then fix the action of $G$ on the state $\Phi$ as follows. We require $G \cdot d \alpha$ to generate a diffeomorphism that in a local flat coordinate system at each point at the boundary $B$ looks like a rotation of angle $d \alpha$ in the plane normal to $B$ (see Fig. \ref{Fig. 5}).

We have to show the existence of such a diffeomorphism. In what follows, we take the existence as given. We will show it in the proof of Theorem \ref{Theorem 3.1} where we will give an explicit construction of $G.$

Then, using Lemma \ref{Lemma 2.2}, we would like to write \eqref{15} as 

\begin{equation} \label{16}
\int_{q_1 (\alpha = 0)=q_{1,B} \atop q(\alpha=\pi, t=\pm iT_E)=\chi}^{q_1(\alpha=2\pi)=q_{1,A}} \operatorname{Vol}^{(B)}(\alpha) e^{i\int_{0}^{2\pi}d \alpha \left( \Theta^{(B)}_{\Phi(\alpha)} \left[\frac{\delta \Phi}{d \alpha} \right] - G[\Phi(\alpha)] \right)}.
\end{equation}

However, before we can finish this conclusion, we have to justify some remaining points.

Note that in \eqref{16}, the slice $\alpha = \pi$ captures the boundary of spacetime $\partial M.$ Hereby, the lower circle passing $t = iT_E$ and the upper circle passing $t = -iT_E$ have been identified (see Fig. \ref{Fig. 4}). No summation over canonical coordinates is lost in this identification. The reason is that the generator $G$ depends, as we will see in the next chapter, only on canonical coordinates in the vicinity of $B$ on the slice $\alpha = \pi.$ 

Our argument so far assumes a specific topology (interval or sphere) for $B.$ What happens if the topology of $B$ is more general?

In Fig. \ref{Fig. 6}, we apply the procedure of Fig. \ref{Fig. 4} to a $\Sigma_1$ with a different topology. 

\begin{figure}[h!]
\centering
  \includegraphics[trim = 0mm 90mm 0mm 20mm, clip, width=0.7\linewidth]{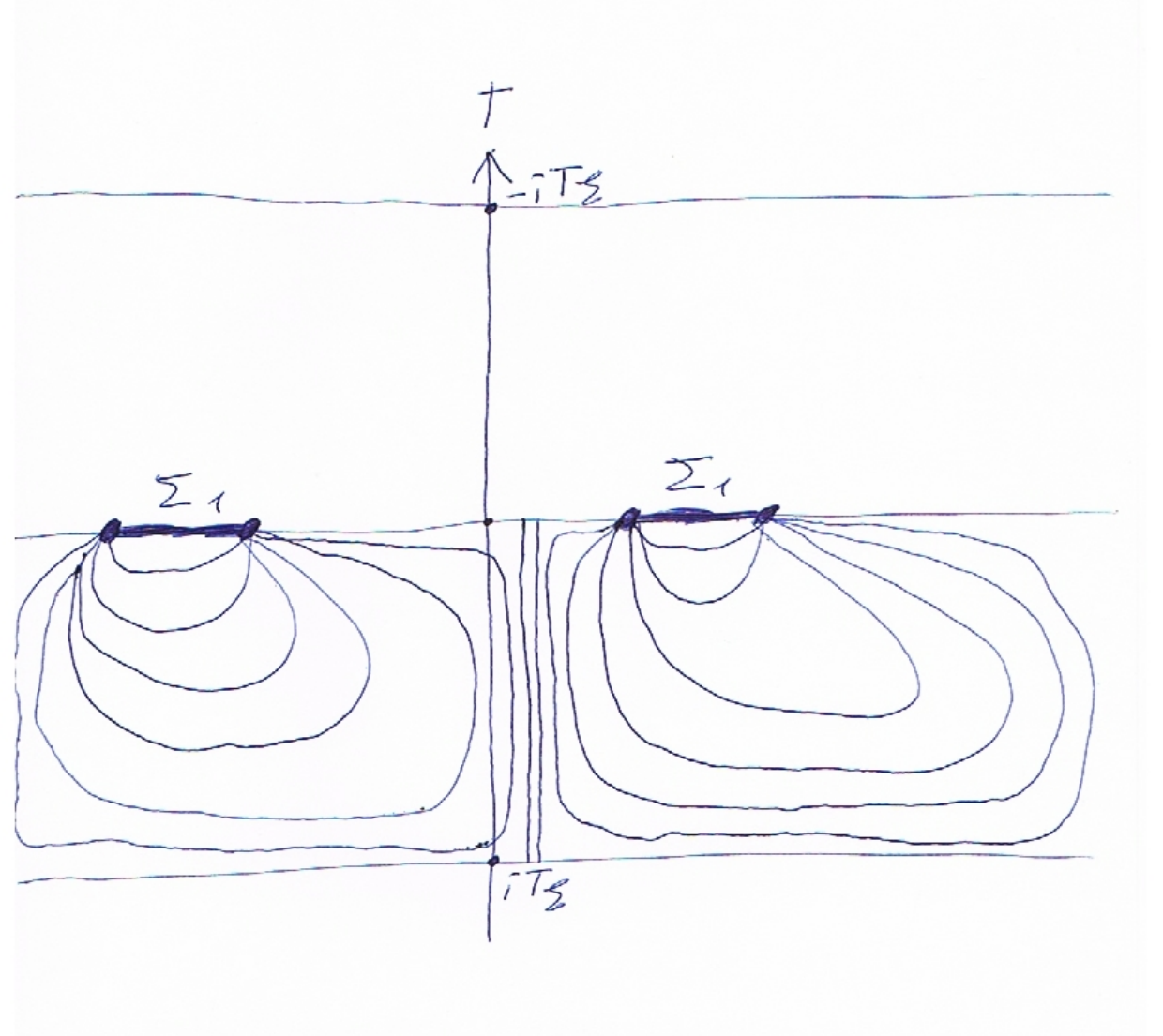}
  \caption{The procedure of Fig. \ref{Fig. 4} for $\Sigma_1$ possessing a different topology.}
	\label{Fig. 6}
\end{figure} 

As shown in Fig. \ref{Fig. 6}, the slice at $\alpha = \pi$ covers parts of the spacetime twice. This is due to the hole in $\Sigma_1.$ However, there occurs no double-counting of canonical coordinates by the same argument as for $\partial M.$ 

Finally, we have specified $G$ by its action as a generator. This determines $G$ only up to a constant. This is because $G$ and $G+c$ generate the same symplectic symmetry for each c-number $c.$ In consequence, \eqref{15} and \eqref{16} are only equal up to a normalization constant. However, $c$ is fixed by the fact that \eqref{15} is a density matrix, i.e. the trace of \eqref{15} is equal to $1.$ Imposing this requirement on \eqref{16} fixes $c$ and ensures the equality of \eqref{15} and \eqref{16} which was to show.   
\end{proof}

Lemma \ref{Lemma 2.3} expresses the reduced density matrix of an arbitrary state $| \Psi \rangle$ as a functional integral involving the generator of the diffeomorphism shown in Fig. \ref{Fig. 4} and Fig. \ref{Fig. 5}. In order to derive the gravitational entropy bound, we apply Lemma \ref{Lemma 2.3} to the particular case of the lowest-energy eigenstate $| \Omega \rangle.$ This is the content of the next lemma.

\begin{lemma} \label{Lemma 2.4}
Let $| \Omega \rangle$ denote the lowest-energy eigenstate. For the possifold-flow $\partial \Sigma \to B,$ there exists a generator $G$ on $(\Gamma^{(B)}, \Theta^{(B)})$ such that the reduced density matrix elements of $| \Omega \rangle$ are given by 

\begin{equation} \label{17}
\begin{split}
&\int \mathcal{D} q_2 ( \langle q_{1,A} | \langle q_2 | ) |\Omega \rangle \langle \Omega | ( |q_{1,B} \rangle |q_2 \rangle ) \\
=& \int_{q_1 (\alpha = 0)=q_{1,B}}^{q_1(\alpha=2\pi)=q_{1,A}} \operatorname{Vol}^{(B)}(\alpha) e^{i\int_{0}^{2\pi}d \alpha \left( \Theta^{(B)}_{\Phi(\alpha)} \left[\frac{\delta \Phi}{d \alpha} \right] - G[\Phi(\alpha)] \right)} \\
=& \langle q_{1,A} | e^{-K} | q_{1,B} \rangle
\end{split}
\end{equation}

with 

\begin{equation}
K = 2 \pi i G.
\label{18}
\end{equation}
\end{lemma}

\begin{proof}
By expansion of the position-eigenstates $| \chi \rangle$ of \eqref{1} in energy eigenstates, it follows that there is an $\alpha \in \mathbb{C} \setminus \{ 0 \}$ such that 

\begin{equation}
\alpha | \Omega \rangle = \lim_{T_E \to -\infty} \int \mathcal{D} \chi e^{HT_E} | \chi \rangle.
\label{19}
\end{equation}

We apply Lemma \ref{Lemma 2.3} to each state $e^{HT_E} | \chi \rangle.$ For each $\chi,$ we obtain the expression \eqref{14} with the generators $G$ on $(\Gamma^{(B)}, \Theta^{(B)})$ only differing by a constant shift. In all cases, $G$ is the generator of the diffeomorphism in Fig. \ref{Fig. 4} and Fig. \ref{Fig. 5}. We can choose an arbitrary $G$ for the moment and let the normalization constant open. Then, there is a $\beta \in \mathbb{C} \setminus \{ 0 \}$ such that according to \eqref{14}, we have

\begin{equation} \label{20}
\begin{split}
&\int \mathcal{D} q_2 ( \langle q_{1,A} | \langle q_2 | ) |\Omega \rangle \langle \Omega | ( |q_{1,B} \rangle |q_2 \rangle ) \\
=& \beta \cdot \lim_{T_E \to -\infty} \int \mathcal{D} \chi \int_{q_1 (\alpha = 0)=q_{1,B} \atop q(\alpha=\pi, t=\pm iT_E)=\chi}^{q_1(\alpha=2\pi)=q_{1,A}} \operatorname{Vol}^{(B)}(\alpha) e^{i\int_{0}^{2\pi}d \alpha \left( \Theta^{(B)}_{\Phi(\alpha)} \left[\frac{\delta \Phi}{d \alpha} \right] - G[\Phi(\alpha)] \right)}.
\end{split}
\end{equation}

The boundary condition at $\alpha = \pi$ is then omitted due to the limes and summation over $\chi$ so that \eqref{20} reads 

\begin{equation}
\beta \int_{q_1 (\alpha = 0)=q_{1,B}}^{q_1(\alpha=2\pi)=q_{1,A}} \operatorname{Vol}^{(B)}(\alpha) e^{i\int_{0}^{2\pi}d \alpha \left( \Theta^{(B)}_{\Phi(\alpha)} \left[\frac{\delta \Phi}{d \alpha} \right] - G[\Phi(\alpha)] \right)}.
\label{21}
\end{equation}

We can set $\beta = 1$ by fixing the shift ambiguity $G \to G + \text{const.}$ in \eqref{21}. In other words, while $G$ is required to generate the diffeomorphism of Fig. \ref{Fig. 4} and Fig. \ref{Fig. 5}, the shift ambiguity is fixed by requiring the trace of \eqref{21} to equal $1.$ This shows the first equality of \eqref{17}. 

The second equality follows from the discussion in Appendix A of \cite{Averin:2024its}. Note that in the right-hand-side of \eqref{17}, $K$ is the quantum-mechanical operator associated to the functional $K$ on $(\Gamma^{(B)}, \Theta^{(B)})$ in \eqref{18}. This association is discussed in detail in Appendix A of \cite{Averin:2024its}.\footnote{As is usual in quantum mechanics, we use the same symbol for the generator and the associated quantum-mechanical operator. This notation is usually no source of confusion as the meaning should be clear from the context.}   
\end{proof}

The statement of Lemma \ref{Lemma 2.4} is that the reduced density matrix of the ground state in a diffeomorphism-invariant field theory is of the form $e^{-K}$ with the operator $K$ constructed in the proof. Such an operator $K$ is often called a modular Hamiltonian. It naturally gives rise to an entropy bound. 

\begin{lemma} \label{Lemma 2.5}
Let $e^{-K}$ be a density matrix on a Hilbert-space $\mathcal{H}$ for some operator $K.$ For $N$ linearly independent eigenstates of $K$ with eigenvalues not bigger than $\lambda,$ one always has the inequality 

\begin{equation}
\ln (N) \le \lambda.
\label{22}
\end{equation}
\end{lemma}

\begin{proof}
For $i=1,\ldots,N$ let $| i \rangle$ denote $N$ orthonormal eigenstates of $K$ with $K | i \rangle = \lambda_i | i \rangle.$ Then, one has due to $\lambda_i \le \lambda$

\begin{equation*}
1 = \operatorname{tr} (e^{-K}) \ge \sum_{i=1}^{N} \langle i | e^{-K} | i \rangle = \sum_{i=1}^{N} e^{-\lambda_i} \ge N e^{-\lambda}.
\end{equation*}  
\end{proof} 

Since the left-hand-side of \eqref{22} looks like an entropy, Lemma \ref{Lemma 2.5} provides an entropy bound. And using Lemma \ref{Lemma 2.4} shows the existence of a non-trivial entropy bound associated to a codimension-2 surface $B$ in any diffeomorphism invariant field theory.

Indeed, a similar argument was used in \cite{Casini:2008cr} to obtain a formulation of the Bekenstein bound. There, a Lorentz-invariant field theory is considered. Using our language, the possifold-flow to a Rindler-wedge $\Sigma_1$ leads to a modular Hamiltonian of the vacuum state which is shown to be a suited Lorentz-transformation generator. This then leads to an entropy bound which captures the qualitative Bekenstein bound.\footnote{There is a further point in the argument. It consists of subtracting suited vacuum contributions in order to obtain a finite statement. In our case, this step will be absent.}

Going back to our case of diffeomorphism invariant theories, the bound in Lemma \ref{Lemma 2.5} is still not satisfactory. The bound \eqref{22} refers to the spectrum of the quantum-mechanical operator $K$ which is usually not known. Rather, we would like to have a bound referring to the functional $K$ on phase space. The left-hand-side of \eqref{22} should then involve a suited phase space volume. Although of little practical use, Lemma \ref{Lemma 2.5} motivates how a suited entropy bound should look like and provides the idea of proof. It is formulated in the next lemma. 

\begin{lemma}[\bf{General Entropy Bound}] \label{Lemma 2.6}
Let $(\Gamma, \Theta)$ be a theory\footnote{with or without Hamiltonian} with a density matrix $e^{-K}$ on the associated Hilbert-space $\mathcal{H}_{\Gamma}.$ The operator $K$ let be associated to the generator $K$ on $(\Gamma, \Theta).$ For $\lambda \in \mathbb{R},$ let $V$ be the phase space volume of the inverse image $K^{-1} ((-\infty, \lambda])$ in $\Gamma.$ One has then

\begin{equation}
\ln (V) \le \lambda.
\label{23}
\end{equation} 
\end{lemma}

\begin{proof}
Similarly to \eqref{17} and \eqref{18}, we introduce the generator $G$ by $K= 2 \pi i G.$ The requirement $e^{-K}$ being a density matrix means 

\begin{equation}
1= \int \operatorname{Vol}(\alpha) e^{i\int_{0}^{2\pi}d \alpha \left( \Theta_{\Phi(\alpha)} \left[\frac{\delta \Phi}{d \alpha} \right] - G[\Phi(\alpha)] \right)}.
\label{24}
\end{equation}

In \eqref{24}, the sum is taken over closed paths $\Phi \colon [0,2\pi] \longrightarrow \Gamma$ in $\Gamma,$ i.e. $\Phi(0)=\Phi(2\pi).$ 

Consider the set $M = X \dot{\cup} Y \dot{\cup} Z$ which is a disjoint union of two open sets $X$ and $Y$ and a possible null set $Z$ in $\Gamma.$ We then have

\begin{equation} \label{25}
\begin{split}
&\int_{\Phi \in M} \operatorname{Vol}(\alpha) e^{i\int_{0}^{2\pi}d \alpha \left( \Theta_{\Phi(\alpha)} \left[\frac{\delta \Phi}{d \alpha} \right] - G[\Phi(\alpha)] \right)} \\
=& \int_{\Phi \in X} \operatorname{Vol}(\alpha) e^{i\int_{0}^{2\pi}d \alpha \left( \Theta_{\Phi(\alpha)} \left[\frac{\delta \Phi}{d \alpha} \right] - G[\Phi(\alpha)] \right)} \\
+& \int_{\Phi \in Y} \operatorname{Vol}(\alpha) e^{i\int_{0}^{2\pi}d \alpha \left( \Theta_{\Phi(\alpha)} \left[\frac{\delta \Phi}{d \alpha} \right] - G[\Phi(\alpha)] \right)}.
\end{split}
\end{equation}

To justify this, we start from the left-hand-side of \eqref{25}. On each time-slice $\alpha,$ we can remove $Z$ from the phase space integral. This does not affect the result since $Z$ is a null set. The paths on the right-hand-side missed compared to the left-hand-side are those passing $Z$ at some time. But those are the ones we just eliminated without changing the result. Hence, we have \eqref{25}. 

Let $N \subseteq \Gamma$ be an open set with finite volume. We can find open sets $U_i \subseteq \Gamma$ such that $N$ is the disjoint union

\begin{equation}
N = \dot{\bigcup_i} U_i \dot{\cup} Z
\label{26}
\end{equation}

with a null set $Z.$ In a moment, we will specify how to choose the sets $U_i.$ In any case, for each i, we choose $\Phi_i \in U_i.$ 

By \eqref{25}, we have

\begin{equation} \label{27}
\begin{split}
&\left| \int_{\Phi \in N} \operatorname{Vol}(\alpha) e^{i\int_{0}^{2\pi}d \alpha \left( \Theta_{\Phi(\alpha)} \left[\frac{\delta \Phi}{d \alpha} \right] - G[\Phi(\alpha)] \right)} \right. \\
& \left. - \int_{N} \operatorname{Vol} e^{-K} \right| \\
=& \left| \sum_i \left( \int_{\Phi \in U_i} \operatorname{Vol}(\alpha) e^{i\int_{0}^{2\pi}d \alpha \left( \Theta_{\Phi(\alpha)} \left[\frac{\delta \Phi}{d \alpha} \right] - G[\Phi(\alpha)] \right)} \right. \right. \\
& \left. \left. - \int_{U_i} \operatorname{Vol} e^{-K} \right) \right| \\
=& \left| \sum_i \left( \int_{\Phi \in U_i} \operatorname{Vol}(\alpha) e^{i\int_{0}^{2\pi}d \alpha \left( \Theta_{\Phi(\alpha)} \left[\frac{\delta \Phi}{d \alpha} \right] - G[\Phi(\alpha)] \right)} \right. \right. \\
&-\int_{U_i} \operatorname{Vol} e^{-K[\Phi_i]} \\
&+ \left. \left. \int_{U_i} \operatorname{Vol} e^{-K[\Phi_i]} - \int_{U_i} \operatorname{Vol} e^{-K[\Phi]} \right) \right|
\end{split}
\end{equation}

\begin{equation*}
\begin{split}
\le & \sum_i \left| \int_{\Phi \in U_i} \operatorname{Vol}(\alpha) e^{i\int_{0}^{2\pi}d \alpha \left( \Theta_{\Phi(\alpha)} \left[\frac{\delta \Phi}{d \alpha} \right] - G[\Phi(\alpha)] \right)} \right. \\
& \left. -\int_{U_i} \operatorname{Vol} e^{-K[\Phi_i]} \right| \\
&+ \sum_i \left| \int_{U_i} \operatorname{Vol} ( e^{-K[\Phi_i]} - e^{-K[\Phi]}) \right| \\
=& \sum_i \left| \int_{\Phi \in U_i} \operatorname{Vol}(\alpha) e^{i\int_{0}^{2\pi}d \alpha \Theta_{\Phi(\alpha)} \left[\frac{\delta \Phi}{d \alpha} \right]} \left( e^{-i \int_{0}^{2 \pi} d \alpha G[\Phi(\alpha)]} - e^{-K[\Phi_i]} \right) \right| \\
&+ \sum_i \left| \int_{U_i} \operatorname{Vol} (e^{-K[\Phi_i]} - e^{-K[\Phi]}) \right| \\
\le & \sum_i \int_{\Phi \in U_i} \operatorname{Vol}(\alpha) \left| e^{-i \int_{0}^{2 \pi} d \alpha G[\Phi(\alpha)]} - e^{-K[\Phi_i]} \right| \\
&+ \sum_i \int_{U_i} \operatorname{Vol} \left| e^{-K[\Phi_i]} - e^{-K[\Phi]} \right|. 
\end{split}
\end{equation*} 

We want to estimate the integrands in the last line of \eqref{27}. We have for any path $\Phi(\alpha)$

\begin{equation*}
\begin{split}
\left| i \int_{0}^{2 \pi} d \alpha G[\Phi(\alpha)] - K[\Phi_i] \right| &= \left| i \int_{0}^{2 \pi} d \alpha G[\Phi(\alpha)] - i \int_{0}^{2 \pi} d \alpha G[\Phi_i] \right| \\
&\le  2 \pi \sup_{\Phi \in U_i} | G[\Phi] - G[\Phi_i] |.
\end{split}
\end{equation*}

Similarly, we have for each $\Phi \in U_i$ 

\begin{equation*}
\begin{split}
| K[\Phi_i] - K[\Phi] | &= 2 \pi | G[\Phi] - G[\Phi_i] | \\
&\le  2 \pi \sup_{\Phi \in U_i} | G[\Phi] - G[\Phi_i] |.
\end{split}
\end{equation*}

From the last two equations, we learn that by choosing the $U_i$ in the decomposition \eqref{26} sufficiently narrow around $\Phi_i,$ we can require the integrands in the last line of \eqref{27} to be arbitrarily small. Let be $\varepsilon > 0.$ We then find a decomposition \eqref{26} with $U_i$ having volumes smaller than $1,$ such that 

\begin{equation} \label{28}
\begin{split}
& \left| \int_{\Phi \in N} \operatorname{Vol}(\alpha) e^{i\int_{0}^{2\pi}d \alpha \left( \Theta_{\Phi(\alpha)} \left[\frac{\delta \Phi}{d \alpha} \right] - G[\Phi(\alpha)] \right)} - \int_{N} \operatorname{Vol} e^{-K} \right| \\
\le & \sum_i \int_{\Phi \in U_i} \operatorname{Vol}(\alpha) \varepsilon + \sum_i \int_{U_i} \operatorname{Vol} \varepsilon \\
\le & \sum_i \int_{U_i} \operatorname{Vol} \cdot 2 \varepsilon = 2 \varepsilon \int_{N} \operatorname{Vol}.
\end{split}
\end{equation}

We can take $\varepsilon \to 0$ and conclude

\begin{equation*}
\int_{\Phi \in N} \operatorname{Vol}(\alpha) e^{i\int_{0}^{2\pi}d \alpha \left( \Theta_{\Phi(\alpha)} \left[\frac{\delta \Phi}{d \alpha} \right] - G[\Phi(\alpha)] \right)} = \int_{N} \operatorname{Vol} e^{-K}
\end{equation*}

for all open $N \subseteq \Gamma$ with finite volume. Using again \eqref{25}, we see that this holds also for $\Gamma.$ Therefore, we have

\begin{equation}
\int \operatorname{Vol}(\alpha) e^{i\int_{0}^{2\pi}d \alpha \left( \Theta_{\Phi(\alpha)} \left[\frac{\delta \Phi}{d \alpha} \right] - G[\Phi(\alpha)] \right)} = \int \operatorname{Vol} e^{-K}
\label{29}
\end{equation}

and according to \eqref{24}

\begin{equation}
\int_{\Gamma} \operatorname{Vol} e^{-K} = 1.
\label{30}
\end{equation}

We conclude

\begin{equation} \label{31}
\begin{split}
1 = \int_{\Gamma} \operatorname{Vol} e^{-K} \ge \int_{K^{-1}((- \infty, \lambda])} \operatorname{Vol} e^{-K}
\ge \int_{K^{-1}((- \infty, \lambda])} \operatorname{Vol} e^{-\lambda} = V e^{-\lambda}
\end{split}
\end{equation}

which shows \eqref{23}. 
\end{proof}

In order to make a quantitative statement of the bound predicted by Lemma \ref{Lemma 2.6} for diffeomorphism invariant theories, we need to construct the generator $K$ of Lemma \ref{Lemma 2.4}. As mentioned, its action as a generator is fixed by Fig. \ref{Fig. 4} and Fig. \ref{Fig. 5}. The remaining ambiguity $K \to K + \text{const.}$ is determined by the normalization of the associated density matrix. We start with the latter point. The next lemma fixes the mentioned ambiguity uniquely. Note that we have restored $\hbar$ for the moment. 

\begin{lemma} \label{Lemma 2.7}
Let $(\Gamma, \Theta)$ be a theory\footnote{with or without Hamiltonian} with a density matrix $e^{-\frac{K}{\hbar}}$ on the associated Hilbert-space $\mathcal{H}_{\Gamma}.$ The operator $K$ let be associated to the generator $K$ on $(\Gamma, \Theta).$ One has then $K \ge 0$ and $0$ is in the image of $K.$
\end{lemma}

\begin{proof}
From \eqref{30}, we have

\begin{equation}
1=\operatorname{tr}_{\mathcal{H}_{\Gamma}} (e^{-\frac{K}{\hbar}})=\int_{\Gamma} \operatorname{Vol} e^{-\frac{1}{\hbar} K[\Phi]}.
\label{32}
\end{equation}

We now consider the limit $\hbar \to 0.$ From \eqref{32}, we conclude $K \ge 0$ on $\Gamma.$ Otherwise, the integral in the last line could be made arbitrarily large contradicting \eqref{32}.

On the other hand, if $K>0,$ the last line in \eqref{32} would vanish for $\hbar \to 0$ again contradicting \eqref{32}. Hence, there is a $\Phi \in \Gamma$ with $K[\Phi]=0.$
\end{proof}

As explained above, it remains to implement the diffeomorphism in Fig. \ref{Fig. 4} and Fig. \ref{Fig. 5} as a generator. This requires using the precise symplectic structure of diffeomorphism invariant field theories. Precisely this, we will do in the next chapter. 

\section{Computation of the Bound}
\label{Kapitel 3}

In the last chapter, we have seen that the gravitational entropy bound is fully determined by a certain generator on the theory's phase space. Our remaining task is therefore to give a precise form of this generator. This requires going into the details of the symplectic structure of diffeomorphism invariant field theories. Such an analysis was done in \cite{Iyer:1994ys} and we will use several results stated there. Especially, we will use the same notations and conventions (in addition to our declarations) and refer to \cite{Iyer:1994ys} for further details.

Remember that we have required $(\Gamma, \Theta, H)$ to be a diffeomorphism invariant field theory with a Lagrange-form $L$ on the spacetime $M.$ The requirement of diffeomorphism invariance and its implications were studied in detail in \cite{Iyer:1994ys}. In particular, according to Lemma 2.1 in chapter 2 of \cite{Iyer:1994ys}, we can take $L$ without loss of generality to be of the form

\begin{equation}
L=L(g_{ab}, R_{cdef}, \nabla_{a_1} R_{cdef}, \ldots, \nabla_{(a_1 \cdots} \nabla_{a_m)} R_{cdef}, \psi, \nabla_{a_1} \psi, \ldots, \nabla_{(a_1 \cdots} \nabla_{a_l)} \psi).
\label{33}
\end{equation}         

That is, the dynamical fields appearing in $L$ are a metric $g_{ab}$ as well as additional fields $\psi$ (which we collectively denote as $\Phi$). We use latin letters for the spacetime indexes. $\nabla$ is the Levi-Civita connection and $R_{abcd}$ the curvature of $g_{ab}.$ 

We are now ready to state and prove the gravitational entropy bound. This is our main result.    

\begin{theorem}[\bf{Gravitational Entropy Bound}] \label{Theorem 3.1}
Let $(\Gamma, \Theta, H)$ be a diffeomorphism invariant field theory on a $n$-dimensional spacetime $M$ foliated by time-slices given by a hypersurface $\Sigma.$ The Lagrange-form is taken to be of the form \eqref{33}. Consider the possifold-flow $\partial \Sigma \to B$ for a codimension-$2$ surface $B.$ On $(\Gamma^{(B)}, \Theta^{(B)}),$ we define the functional 

\begin{equation}
K[\Phi] = 2\pi \oint_B X^{cd} \varepsilon_{cd} + c
\label{34}
\end{equation}

for a state $\Phi \in \Gamma^{(B)}.$ The $(n-2)$-form $X^{cd}$ is given by

\begin{equation}
(X^{cd})_{c_3 \cdots c_n} = -E_R^{abcd} \varepsilon_{abc_3 \cdots c_n}
\label{35}
\end{equation}

where $E_R^{abcd}$ is defined as

\begin{equation}
\delta_R L = \varepsilon E_R^{abcd} \delta R_{abcd}
\label{36}
\end{equation}

and $\delta_R$ is understood as a variation of the curvature tensor viewed as an independent field with $E_R^{abcd}$ required to have the same algebraic tensor symmetries as $R_{abcd}.$ $\varepsilon$ denotes the volume $n$-form. 

$\varepsilon_{cd}$ denotes the binormal to $B.$ That is, $\varepsilon_{cd}$ is the natural volume element on the tangent space perpendicular to $B$ with the orientation given by $\varepsilon_{cd} T^c X^d > 0.$ $T^a$ is hereby a future-directed timelike vector determining the orientation of $\Sigma.$ $X^a$ is a spacelike vector pointing outwards to $B$ and determines the orientation of $B.$

The constant $c \in \mathbb{R}$ is determined by the requirement $K \ge 0$ and $0$ to be in the image of K.

For $\frac{\lambda}{\hbar} \in \mathbb{R},$ let $V$ be the phase space volume of the inverse image $K^{-1} ((-\infty, \lambda])$ in $\Gamma^{(B)}.$ One has then

\begin{equation}
\ln (V) \le \frac{\lambda}{\hbar}.
\label{37}
\end{equation} 
\end{theorem}

\begin{proof}
We have to show that the diffeomorphism of Fig. \ref{Fig. 4} and Fig. \ref{Fig. 5} used in the constructive proof of Lemma \ref{Lemma 2.3} leads in Lemma \ref{Lemma 2.4} to the generator $K$ of the form \eqref{34}. The assertion then follows from Lemma \ref{Lemma 2.6} and \ref{Lemma 2.7}. 

For any point in $B,$ we can find a coordinate neighborhood with coordinates $(x^A, y^i)$ where $\frac{\partial}{\partial y^i}$ are required to be tangential to $B$ at $B$ while $\frac{\partial}{\partial x^A}$ are not. We use capital letters $A,B,\ldots = 1, 2$ and lowercase letters $i,j,\ldots = 3, \ldots, n$ in the middle of the alphabet to label the corresponding indexes. 

By imposing gauge-fixing, we can require for any state in $\Gamma$ the metric near $B$ to be of the form

\begin{equation} \label{38}
\begin{split}
ds^2 &= g_{ab} dx^a dx^b \\
&= \eta_{AB} dx^A dx^B + g_{ij} dy^i dy^j + \ldots
\end{split}
\end{equation}

with $\eta$ meaning the flat metric. $g_{ij}$ is state-dependent and the dots are vanishing on $B.$ 

We choose coordinates $x^A = (t,x)$ such that in Lorentz-signature 

\begin{equation}
\eta_{AB} dx^A dx^B = -dt^2 + dx^2.
\label{39}
\end{equation}

We require the coordinate system $(t,x,y^i)$ to be positively oriented. That is, $\partial_x$ is pointing outward on $B$ and $\partial_t$ is future-directed towards the evolution generated by $H.$

As in Fig. \ref{Fig. 4} and Fig. \ref{Fig. 5}, we consider the state in Euclidean signature by analytic continuation

\begin{equation}
t=it_E.
\label{40}
\end{equation}

We can then easily describe the diffeomorphism in Fig. \ref{Fig. 5} using the polar coordinates

\begin{equation} \label{41}
\begin{split}
t_E &= r \sin (\varphi) \\
x &= r \cos (\varphi)
\end{split}
\end{equation}

where it takes the form in accordance with orientation

\begin{equation}
-d \alpha \frac{\partial}{\partial \varphi} =: d \alpha \xi
\label{42}
\end{equation}

and we have defined the vectorfield $\xi = -\partial_\varphi.$

Before we can proceed to calculate the generator of \eqref{42}, we need to justify that \eqref{42} is well-defined. 

There are several points to clarify. First, \eqref{42} is to act on $\Gamma^{(B)}.$ Its elements can be thought as being parametrized by solutions of the equations of motion. The field theory $(\Gamma, \Theta, H)$ is diffeomorphism invariant and may possess additional gauge symmetries. After imposing gauge-fixing and gauge constraints, there remain free initial conditions on $\Sigma$ forming the parametrization of $\Gamma$ (and the associated $\Gamma^{(B)}$ via the possifold-flow $\partial \Sigma \to B$). However, in imposing this requirements, there may still be free degrees of freedom playing the role of boundary values. The information about their values can be thought as being contained in a quantity $T$ at $B$ where $T$ is formed out of the dynamical fields $\Phi$ and their derivatives. The information whether these degrees of freedom are fixed as boundary conditions or appear as free coordinates to be summed over in the functional integral is contained in $\Theta.$ As discussed in detail in chapter 2 of \cite{Averin:2024its}, the presymplectic potential $\theta$ determines the canonical coordinates to be included in a possifold-flow $\partial \Sigma \to B$ of a field theory. 

Whether these degrees of freedom are fixed or appear as canonical variables to be summed over in the functional integral, in each case \eqref{42} is well-defined. Since $T$ has a well-defined value at $B$ for $r=0$ in the coordinates \eqref{41}, we have $\mathcal{L}_\xi T = 0.$ Hence, $T$ is unchanged by \eqref{42}.\footnote{We will see in a moment that by the same argument, the entropy bound takes in each case the same form. This is a hint, that the mentioned degrees of freedom are needed for saturation of the gravitational entropy bound. Hence, this suggests, they are the relevant degrees of freedom responsible for black hole entropy and microstates. This is precisely the proposal made in \cite{Averin:2018owq,Averin:2019zsi} which led to the possifold notion introduced in \cite{Averin:2024its}.}  

The second point is to justify that \eqref{42} is globally well-defined over $B.$ In \eqref{42}, we have used a local coordinate neighborhood of a given point around $B.$ Consider an observer in a point at $B$ in the intersection of two such coordinate neighborhoods. For the observer, \eqref{42} describes in both coordinate systems a rotation in the locally flat coordinates perpendicular to $B.$ Since the observer's locally flat coordinate systems are related by a rotation, the rotation angle is equal in both coordinate systems. Hence, the diffeomorphisms constructed in both neighborhoods agree for the observer locally. Therefore, \eqref{42} uniquely defines a diffeomorphism near $B.$ As we will see in a moment, its form in the interior of $B$ does not affect the associated generator. 

According to Lemma \ref{Lemma 2.4}, the functional $K$ to be constructed is given by \eqref{18}. Hereby, $G$ is according to the constructive proof of Lemma \ref{Lemma 2.3} the generator of the diffeomorphism $\xi.$ By equation \eqref{31} (note also \eqref{7} for the correct sign) of \cite{Averin:2024its}, this means the requirement

\begin{equation}
\delta G = \Omega^{(B)} [\delta \Phi, \mathcal{L}_\xi \Phi].
\label{43}
\end{equation}

Equation (75) in \cite{Iyer:1994ys} gives for the presymplectic current of any diffeomorphism invariant field theory

\begin{equation}
\omega_\Phi [\delta \Phi, \mathcal{L}_\xi \Phi] = d( \delta Q_\xi - \xi \cdot \theta)
\label{44}
\end{equation}

with the Noether-charge $(n-2)$-form $Q_\xi$ associated to an arbitrary spacetime vectorfield $\xi.$ The latter is according to Proposition 4.1 in chapter 4 of \cite{Iyer:1994ys} of the form

\begin{equation}
Q_\xi [\Phi] = W_c(\Phi) \xi^c + X^{cd}(\Phi) \nabla_{[c} \xi_{d]} + Y(\Phi, \mathcal{L}_\xi \Phi) + dZ(\Phi,\xi).
\label{45}
\end{equation}

The coefficients $W,X,Y,Z$ are functions of the shown fields and their derivatives. $Y$ is linear in $\mathcal{L}_\xi \Phi$ and $Z$ is linear in $\xi.$ The decomposition \eqref{45} is not unique but it can always be chosen such that $X$ is given by \eqref{35}. 

We can then evaluate \eqref{43} using \eqref{44} as

\begin{equation} \label{46}
\begin{split}
\delta G &= \oint_{B} (\delta Q_\xi - \xi \cdot \theta) \\
&= \delta \left( \oint_B Q_\xi \right) \\
&= \delta \left( \oint_B X^{cd}(\Phi) \nabla_{[c} \xi_{d]} \right)
\end{split}
\end{equation}

where the second equality holds because $\xi$ vanishes on $B.$ 

From the first general equality, we see that the generator of a diffeomorphism - if existent - only depends on the value of the diffeomorphism $\xi$ and dynamical fields $\Phi$ near the boundary $B.$ This property was used in the proof of Lemma \ref{Lemma 2.3} and is now justified retrospectively.

In the third equality, we have used the expression \eqref{45} for the Noether-charge. Only the second term in \eqref{45} contributes. The first term does not contribute because $\xi$ vanishes on $B.$ The third term does not contribute because $\mathcal{L}_\xi \Phi |_B =0$ as $\Phi$ has well-defined values at $B.$ 

With respect to the coordinates used in \eqref{39}, the vectorfield $\xi$ can be written as

\begin{equation}
\xi = -i(x \partial_t + t \partial_x).
\label{47}
\end{equation}

Hence, we find on $B$

\begin{equation}
\nabla_{[c} \xi_{d]} = -i \varepsilon_{cd}
\label{48}
\end{equation}

with the binormal $\varepsilon_{cd}$ as in the assertion.

The requirement \eqref{46} on $G$ then reads

\begin{equation}
\delta G = -i \delta \left( \oint_B X^{cd}(\Phi) \varepsilon_{cd} \right)
\label{49}
\end{equation}

which directly proves the existence of $G.$ The form of $K$ then follows from \eqref{18} and this proves the assertion.     
\end{proof}

Theorem \ref{Theorem 3.1} gives an entropy bound for any diffeomorphism invariant field theory. The explicit construction prescribed by Theorem \ref{Theorem 3.1} is illustrated in the following example.

\begin{example} \label{Example 3.1}
We consider a diffeomorphism invariant field theory on a $d$-dimensional spacetime as required by Theorem \ref{Theorem 3.1}. We take the Lagrange-form \eqref{33} to be

\begin{equation}
L = \frac{1}{16 \pi G} \varepsilon (R-2\Lambda) + L_m \varepsilon.
\label{50}
\end{equation}

That is, we consider the Einstein-Hilbert Lagrangian with Newton's constant $G$ and allow for a cosmological constant $\Lambda.$ Furthermore, we require additional matter-fields $\psi$ to be minimally coupled

\begin{equation}
L_m = L_m(g_{ab}, \psi, \nabla_{a_1} \psi, \ldots, \nabla_{(a_1 \cdots} \nabla_{a_l)} \psi),
\label{51}
\end{equation}

i.e. there is no explicit dependence on the curvature tensor. In other words, this requirement means the validity of the equivalence principle. 

We take a $(d-2)$-dimensional surface $B$ as in Theorem \ref{Theorem 3.1}. Furthermore, we use a coordinate system $(t,x,y^1,\ldots,y^{d-2})$ as in \eqref{38} and \eqref{39}. 

We find for \eqref{36}

\begin{equation}
\delta_R L = -\frac{1}{16 \pi G} g^{ad} g^{bc} \varepsilon \delta R_{abcd}.
\label{52}
\end{equation}

Note that there is no contribution from the matter term or the cosmological constant in \eqref{50}. Hence, using \eqref{35} we obtain 

\begin{equation}
X^{cd} \varepsilon_{cd} = \frac{1}{16 \pi G} g^{ad} g^{bc} \varepsilon_{abc_3 \cdots c_d} \varepsilon_{cd}.
\label{53}
\end{equation}

Therefore, the integrand in \eqref{34} is

\begin{equation} \label{54}
\begin{split}
& 2\pi (X^{cd})_{y^1 \cdots y^{d-2}} \varepsilon_{cd} dy^1 \wedge \ldots \wedge dy^{d-2} \\
=& 2\pi \cdot \frac{1}{16 \pi G} g^{ad} g^{bc} \varepsilon_{aby^1 \cdots y^{d-2}} \varepsilon_{cd} dy^1 \wedge \ldots \wedge dy^{d-2} \\
=& \frac{1}{4G} \sqrt{\det (g_{ij})} dy^1 \wedge \ldots \wedge dy^{d-2}
\end{split}
\end{equation}

where the last equality follows from using the explicit form \eqref{38} and \eqref{39} at $B.$ Hence, we obtain for \eqref{34} by integrating over $B$ 

\begin{equation}
K = \frac{A}{4G}
\label{55}
\end{equation}

with $A$ denoting the area of $B$ in a particular state. The constant appearing in \eqref{34} has been fixed to $c=0$ by the requirements of Theorem \ref{Theorem 3.1}. From Theorem \ref{Theorem 3.1}, we then immediately conclude the following
\end{example}

\begin{corollary} \label{Corollary 3.1}
Let $(\Gamma, \Theta, H)$ be a diffeomorphism invariant field theory on a $d$-dimensional spacetime foliated by time-slices given by a hypersurface $\Sigma.$ The Lagrange-form is taken to be of the Einstein-Hilbert form \eqref{50} and \eqref{51}. Consider the possifold-flow $\partial \Sigma \to B$ for a codimension-$2$ surface $B.$ Let $V$ be the phase space volume of the set of states in $(\Gamma^{(B)}, \Theta^{(B)})$ with area of $B$ not bigger than a fixed $A.$ One has then

\begin{equation}
\ln (V) \le \frac{A}{4 \hbar G}.
\label{56}
\end{equation}  
\end{corollary}

\section{Discussion} 
\label{Kapitel 4}

Theorem \ref{Theorem 3.1} essentially states that the states with fixed value of the quantity \eqref{34} are contained in a finite volume of phase space. As, for instance, explained in chapter 4.3 of \cite{Averin:2024its}, this is expected from qualitative arguments involving black hole thermodynamics.

Indeed, the quantity \eqref{34} is known to appear in the context of black holes. Based on \cite{Wald:1993nt}, the first law of black hole mechanics was shown to hold in \cite{Iyer:1994ys} for any stationary black hole solution with bifurcate Killing-horizon. The quantity \eqref{34} plays hereby the role of the black hole entropy when evaluated on the black hole state with $B$ being the bifurcation surface.

Our contribution here is that the functional \eqref{34} has a much more general meaning. It appears in the general functional integral expression for reduced density matrices in diffeomorphism invariant theories as we showed in Lemma \ref{Lemma 2.3}. As stated in Lemma \ref{Lemma 2.4}, for the ground state, it plays the role of the modular Hamiltonian implying several entropy bounds. Lemma \ref{Lemma 2.5} gives an entropy bound for the quantum-mechanical operator associated to \eqref{34}. Theorem \ref{Theorem 3.1} gives a statement about \eqref{34} itself.

Our results here materialize in part the program outlined in \cite{Averin:2024its}. We refer to chapter 1 there for reviewing the program and chapter 4.3 for placing our results here in context and their meaning for future analysis. Especially, the need for a precise formulation of a gravitational entropy bound was pointed out there, as such a bound is a statement about quantum-mechanical sensitivity. We will analyze how the gravitational entropy bound restricts the sensitivity of observables and implies several other conjectured properties of quantum gravity at a different point. Remarkably, note that the equivalence principle implies \eqref{56} which suggests a lower sensitivity on short length scales than one might expect. 

As noted, the functional \eqref{34} appears in the expression for the reduced density matrix in Lemma \ref{Lemma 2.3}. It is then tempting to ask for a generalization of the Ryu-Takayanagi formula, i.e. an expression for the associated entanglement entropy. Indeed, our derivation shares some similarities with proofs of the Ryu-Takayanagi formula presented in \cite{Rangamani:2016dms,Dong:2013qoa,Dong:2017xht}. We will analyze this further at a different point.   

\section*{Acknowledgements}
We thank Alexander Gußmann for many discussions on this and other topics in physics.

\end{document}